\documentclass[amsthm]{article}
\usepackage[utf8]{inputenc}
\usepackage[T1]{fontenc}
\usepackage[numbers]{natbib}
\usepackage{multirow}
\usepackage{enumerate}
\usepackage{booktabs}
\usepackage{url}
\usepackage{algorithm}
\usepackage{algpseudocode}
\usepackage{alltt}
\usepackage{mesmaths}
\usepackage{tikz-cd}
\usepackage[hidelinks]{hyperref}
\usepackage{orcidlink}

\let\Input=\Require
\let\Output=\Ensure
\newcounter{code}
\makeatletter
\newenvironment{code}[1][h!tb]{%
  \let\c@algorithm\c@code
  \renewcommand{\ALG@name}{Code}
  \begin{algorithm}[#1]%
  }{\end{algorithm}
}
\makeatother

\PassOptionsToPackage{hidelinks}{hyperref}

\title{Computing euler products and coefficients of classical
modular forms for twisted L-functions}
\author{Pascal Molin\footnote{Université Paris Cité and Sorbonne Université,
CNRS, INRIA, IMJ-PRG, F75013 Paris, France, pascal.molin@imj-prg.fr \orcidlink{0000-0003-2462-8751}}}

\begin{document}

\maketitle
\thispagestyle{empty}
\enlargethispage{2em}

\begin{abstract}
    We describe a complete algorithm to compute millions
    of coefficients of classical modular forms in a few seconds.
    We also review operations on Euler products and
    illustrate our methods with a computation of triple
    product L-function of large conductor.
\end{abstract}

\tableofcontents

\section{Introduction}

We describe an algorithm allowing practical computation of
many coefficients of classical modular forms over $\mathrm{GL}_2$.

More specifically, we are interested in computing $n$ Fourier coefficients
of a classical eigenform $f(z)=\sum a_nq^n\in S_2(N)$ for small level
$N < 1000$ and $n$ in the range $10^6$ to $10^9$.

This kind of range may seem unusual: only $O(N)$ Fourier coefficients
are needed to identify a modular form, and $O(\sqrt N)$ to compute
its L-function. However the necessity to compute
a number of Fourier coefficients much larger than the level $N$ appears
when one computes L-functions involving twists, such as Rankin-Selberg
products or twists by Hecke characters or Artin representations.

This work is actually a preparatory step for computations inspired
by Merel and Lecouturier, who are in demand of
many special values $L(f\otimes \rho,s)$ of L-functions of modular forms
twisted by Artin representations.

Meanwhile we give another application of our method with the computation
of a triple product L-functions $L(f\otimes g\otimes h,s)$,
where $f,g,h$ are three modular forms.

\subsection{Strategy}

The strategy we use is to express modular forms in terms of Eisenstein series,
\emph{à la} Borisov-Gunnels \cite{BorisovGunnels}.

The principle is that, except for a well understood obstruction,
the vector space $M_k(N,\chi)$ should be generated
by Eisenstein series and \emph{products} of Eisenstein
series of lower weights $E_\ell^{\phi,\psi}$, $\ell\leq k$.

The Fourier coefficients of Eisenstein series are completely explicit
\[
  E_k(\phi,\psi) = e_k^{\phi,\psi}
  + \sum_{n\geq 1} \bigl(\sum_{d\mid n}\phi(n/d)\psi(d)d^{k-1}\bigr) q^n
\]
with a constant term $e_k^{\phi,\psi}$ that can be expressed
as a Bernoulli coefficient, see Equation \eqref{eq:e0}.

The coefficients of a product can be obtained using fast multiplication
of power series, in $O(n\log n)$ operations for $n$ coefficients.

Now for any particular modular form $f\in M_k(N,\chi)$, assume
we have found by other means an expression
\begin{equation}\label{eq:mfbg}
    f = \sum_{i=1}^r c_i E_i E_i'
\end{equation}
where $E_i$ and $E_i'$ are Eisenstein series, and the coefficients $c_i$
belong to the coefficient ring $\Z[E_i,E_i',f]$:
we obtain a $O(n\log n)$
algorithm for the computation of $n$ coefficients $a_n(f)$.

\subsection{Contributions}

We claim no originality in the overall idea which motivated
recent works on Borisov-Gunnels decompositions
\cite{DicksonNeururer, BelabasCohen}, and which underlies
some advanced modular form functionality in Pari/GP\cite{PariGP}.
In particular this decomposition is used to compute Fourier
expansions at cusps different from $\infty$.
For the cusp $\infty$ a proof of concept by D. Loeffler is
also mentionned in \cite{DicksonNeururer} and a general Sage package can
be found at \cite{NeururerGithub}.
\medskip

However the timings obtained are orders of magnitude away from our work.

\subsubsection{Eisenstein series}

In fact in order to obtain a really efficient method we were led to
revisit Fourier expansion of Eisenstein series, to prevent this
seamingly easy step to bottleneck the whole process.

We obtain an impressive speedup over existing methods. 

\begin{thm}\label{thm:eiscost}
    Let $\phi,\psi$ be Dirichlet characters of order $o_\phi,o_\psi$
    modulo $N_\phi,N_\psi$,
    with values in a ring $R[\phi,\psi]$, and
    $E_k^{\phi,\psi}=\sum_n a_n q^n$ the Eisenstein series of weight $k$
    associated to $\phi,\psi$.
    For $n>1$ we can compute the first $n$ coefficients
    \[
        a_n
        = \sum_{d\mid n}\phi(n/d)\psi(d)d^{k-1} \in R[\phi,\psi]
    \]
    of index $n\geq1$
    with \textbf{less than} $n+o_\phi+o_\psi$
    additions and multiplications in $R[\phi,\psi]$,
    and $O(n+N_\psi+N_\psi)$ operations on indices.

    Furthermore if $R[\phi,\psi]$ is a finite field the computation
    can be done with $O(n+N_\phi+N_\psi)$ space.
\end{thm}

We will even prove that in the very useful case of weight $k=1$
we need strictly less than $n$ multiplications in $R$.

\subsubsection{Modular forms}

The last ingredient for coefficients of modular forms is a fast
implementation of product of power series. We
leverage on Flint's highly optimized Fast Fourier Transform product
over small primes \cite{Flint} for this last step.
We also take advantage of the fact that it
suffices to do all computations modulo an FFT prime to obtain the
coefficients of an eigenform, thanks to Hasse bound on coefficients.

\begin{thm}\label{thm:mfcost}
    Let $f\in S_k(N)$ be a modular form, and assume $f$ is given
    with an explicit decomposition \eqref{eq:mfbg}.
    We can compute the first $n$ Fourier coefficients
    $a_n(f)\in\Z[f]$ with $O(krn\log n)$ binary operations
    and $O(kn)$ memory.
\end{thm}

\subsubsection{General operations on L series}

The key improvements mentionned above are obtained from a
global reflection on the algorithmics of Dirichlet series.

In particular, and this is important for our application on
triple product L-function, we consider the problem of computing
L-series of tensor products and symmetric powers.

We describe an algorithm of optimal complexity, where
the precise definitions will be given in Section \ref{sec:4products}.
\begin{thm}\label{thm:prodsymcost}
    Let $f,g$ be two arithmetic objects whose L-functions
    $L(f,s)=\sum_{n\geq 1}a_n(f) n^{-s}$
    and $L(g,s)=\sum_{n\geq 1}a_n(g)n^{-s}$ admit a Euler product.
    Then from the knowledge of the first $n$ coefficients
    of $f$ and $g$ we can compute the first $n$ coefficients
    of $L(f\oplus g,s)$ and $L(f \otimes g,s)$ with $O(n)$ operations.
    From the first $nk$ coefficients of $f$ we can compute the first
    $n$ coefficients of $L(\Sym^k f,s)$ with $O(nk)$ operations.
\end{thm}

\subsection{Implementation}

We wrote a small program demonstrating the computation of coefficients
for a choice of modular forms: we obtain millions of coefficients
in seconds, and for lengths $n\leq 10^9$ we need only a few minutes
(see Section \ref{sec:timings}).

Our program is available as an example
file \texttt{example/mfcoefs} in the repository of Flint \cite{Flint},
so that the timings can be easily reproduced.

For larger lengths the method is still of interest, but we start to run
into memory issues and need other FFT code
(e.g. for $n=10^{10}$, storing $2^{34}$ coefficients
as 64 bits integers occupies 128GB and we should adopt more segmented
approaches).

\subsection{Other methods}

Default modular form packages in Magma \cite{Magma}, Pari/GP \cite{PariGP}
and Sage \cite{sagemath} have
not been designed to compute a large number of coefficients.
We make a small performance benchmark in Section\ref{sec:compare}
to emphasize this fact.

Yet there are other ideas to compute either a single coefficient
$a_p(f)$ for large $p$ or a bunch of these:
\begin{itemize}
    \item Mascot proposed another quasilinear algorithm in
        \cite[3.1]{Mascot2013}, based on modular equations and
        Newton iteration.
    \item The Galois representation method \cite{CouveignesEdixhoven, Mascot2013}
        gives an algorithm to compute any coefficient $a_p(f)$ in time
        polynomial in $\log(p)$.
    \item When the modular form is associated to a rational
        elliptic curve we have efficient point counting algorithms
        (a very special case of the above).
\end{itemize}
To our knowledge, with the notable exception of weight $k=1$ \cite{Lauder},
there is no efficient and publicly available
implementation of the generic Galois representation method.
Moreover it seems that our method
remains competitive for moderate inputs, as we demonstrate in
Section \ref{sec:compare} in the case of a small elliptic curve L-series.

One can also make use of other decompositions of modular forms into
easily computable functions, like theta series and eta products.
It seems the Eisenstein series expression is the most
general.


\section{Operations on Dirichlet series}

\subsection{Dirichlet coefficients}

Let $f$ be an arithmetic object admitting a L-function,
we denote by $a_n(f)$ the coefficients of its Dirichlet series
\[
  L(f,s) = \sum_n a_n(f)n^{-s}.
\]

When this $L$-function admits a Euler product we have
\[
  L(f,s) = \prod_p F_p(f,p^{-s})^{-1}
\]
for Euler polynomials of the form
\[
  F_p(f,T) = \prod_{i=1}^d (1-\alpha_{p,i}(f)T)
\]
and whose roots satisfy Ramanujan bound
$\abs{\alpha_{p,i}(f)}\leq p^{w/2}$ for a fixed motivic
weight $w$, with equality when $p\nmid N$.
In particular almost all factors (those for $p\nmid N$) have common
degree $d$ which is the degree of the L-function.

Finally, when $f$ corresponds to an automorphic theta function or
modular form (for degree $d=1$ and $d=2$), we have a Fourier expansion
\[
  f(z) = \sum_n a_n(f) q^n
\]
and the L-function is obtained as the Mellin transform of $f$.

The modular form series expansion is the same as the Dirichlet series.

\subsection{A ring with four products}\label{sec:4products}

Let $f,g$ be two such arithmetic objects, we can consider four products
on the associated Dirichlet series, which we denote according to the
Artin formalism of L-functions.

\begin{enumerate}
  \item 
      the product of powers series $a_n(fg) = \sum_{i+j=n}a_i(f)a_j(g)$
  \item 
      the Dirichlet convolution $a_n(f\oplus g) = \sum_{ij=n}a_i(f)a_j(g)$
  \item 
      if $f$ and $g$ admit a Euler product,
      the tensor product $a_n(f\otimes g)$ which is defined on Euler factors by
      \[
          F_p(f\otimes g,T) = \prod_{i,j}(1-\alpha_{p,i}(f)\alpha_{p,j}(g) T)
      \]
  \item
      the pointwise product of coefficients $a_n(f\odot g)=a_n(f)a_n(g)$.
\end{enumerate}

\begin{rem}
When $f$ or $g$ have an Euler expansion, one of which has degree $d=1$
the last two coincide
\[
    a_n(f\otimes g) = a_n(f)a_n(g).
\]
The case of larger degree corresponds
to arithmetic objects whose Dirichlet coefficients are in fact indexed by
higher rank modules such as integral ideals, and whose arithmetic
is adequately captured at the level of degree $d$ Euler polynomials.

We will detail below (Section \ref{sec:rankinselberg}) the precise
relation between the pointwise product and the tensor product. We
can ignore the former for arithmetic applications.
\end{rem}

\begin{rem}
  We define these products at the level of Dirichlet series.
  We \emph{do not} pretend that these operations produce
  actual L-functions: we need to adjust the
  Euler factors in a subtle arithmetic way at primes dividing
  the conductor in order
  to obtain a nice symmetric functional equation.
\end{rem}

To each of these products corresponds a natural
pointwise product in the appropriate setting:

\begin{enumerate}
\item
  $(fg)(z) = f(z)g(z)$ for the evaluation of modular forms
\item
  $L(f\oplus g,s) = L(f,s)L(g,s)$ for the evaluation of L functions.
  When the L functions admit a Euler product
  the multiplication is also valid at the level of Euler
  factors $F_p(f\oplus g,T) = F_p(f,T)F_p(g,T)$.
\item
    $\set{\alpha_{p,i}(f\otimes g)} = \set{\alpha_{p,i}(f)}\times\set{\alpha_{p,i}(g)}$
  for the roots of Euler factors
  (this is in fact the tensor product of representations
  which corresponds to multiplying the eigenvalues).
\end{enumerate}

Our point in the remaining of this section is the following:
Theorem \ref{thm:prodsymcost} means that \emph{the only expensive
operation in this ring of Dirichlet series is the usual product
of series $fg$, whose complexity is $O(n\log n)$.}

Another way of viewing it is that we always implement these operations in
the efficient pointwise setting, but 
we have $O(n)$ conversion algorithms
between Euler factor data and Dirichlet coefficients data
(this will be Theorem \ref{thm:eulercomplexity}),
while conversions between coefficients $(a_i)_{i\leq n}$ and evaluated
values $f(q_i)_{i=1..n}$ is a $O(n\log n)$ process in the most favorable
setting of structured $q_i$ being roots of unity
(the Fast Fourier Transform algorithm).

\subsection{Operations on Euler factors}

\subsubsection{Products and powers}

\begin{defn}\label{def:poloperations}
    Let $P(T)=\prod_{i=1}^{d} (1-\alpha_i T)$ and
        $Q(T)=\prod_{j=1}^{e} (1-\beta_j T)$ be two polynomials.
    We consider
    \begin{enumerate}
    \item the usual product
    $(P\oplus Q)(T) = P(T)Q(T) = \prod_{i,j}(1-\alpha_iT)(1-\beta_j T)$
    \item the tensor product
    $(P\otimes Q)(T) = \prod_{i,j}(1-\alpha_i\beta_j T)$
    \item the $k$-th root power
    $P^{\circ k}(T) = \prod_i(1-\alpha_i^k T)$
    \item the $k$-th symmetric power
    $(\Sym^k P)(T) = \prod_{1\leq i_1\leq\dots \leq i_k\leq d}
        (1-\alpha_{i_1}\dots \alpha_{i_k} T)$.
    \end{enumerate}
\end{defn}

%
The rest of this section is devoted to prove the following.
\begin{thm}\label{thm:poloperations} 
    With the notations of Definition \ref{def:poloperations},
    for any precision $\ell>0$, if $\ell!$ can be inverted in $R$
    we can compute the first $\ell$ coefficients
    of the polynomials $(P\oplus Q)(T)$ and
    $(P\otimes Q)(T)$ with $O(\ell^2)$ operations in $R$,
    and those of $\Sym^k P(T)$ with $O(k\ell(k+\ell+d))$ operations in $R$.
\end{thm}

\begin{rem}
We do not use the expression
    $(P\otimes Q)(T) = \Res_U(P(U),U^eQ(T/U))$
as a resultant.
In fact the method we will describe is simpler
to implement, it makes it easier to
compute only part of the product (a key feature for L-functions),
and it is less sensitive to the nature of the ring
of coefficients (e.g. inexact).
\end{rem}

\subsubsection{Symmetric polynomials}

Let $\alpha_1,\dots \alpha_d$ be $n$ variables, we consider
the degree $d$ polynomial
\[
    P(T) = \prod_{i=1}^d (1-\alpha_iT)
\]
whose coefficients are given by the elementary
symmetric polynomials
\begin{equation}
    \sigma_k(P) =
    \sum_{1\leq i_1<i_2<\dots <i_k \leq d}
    \alpha_{i_1}\dots \alpha_{i_k}
\end{equation}

We also consider the complete homogeneous sums of degree $k$
\begin{equation}
    h_k(P) =
    \sum_{1\leq i_1\leq i_2\leq \dots \leq i_k\leq d}
    \alpha_{i_1}\dots \alpha_{i_k}
\end{equation}
and the power sums of Newton
\begin{equation}
    N_k(P) = x_1^k+\dots+x_d^k.
\end{equation}

By an easy counting argument we have
\begin{lem}
    Assume $\abs{\alpha_i}=p^{w/2}$ for $i=1\dots d$, then
    we have
    $\abs{\sigma_k(P)}\leq \binom{d}{k} p^{\frac{kw}2}$,
    $\abs{h_k(P)}\leq \binom{d+k}{k} p^{\frac{kw}2}$ and
    $\abs{N_k(P)}\leq d p^{\frac{kw}2}$.
\end{lem}

In particular in the context of L-functions we can
expect $\abs{N_k(P)}\leq \abs{\sigma_k(P)}$
for $k<d$.

\subsubsection{Conversions}

The links between these three families are convieniently obtained
via their generating series. We have
\begin{equation}\label{eq:serhk}
    H(P,T)
    = \sum_{k\geq 0} h_k(P) T^k
    = \frac1{P(T)}
\end{equation}
and
\begin{equation}\label{eq:serNk}
    N(P,T)
    = \sum_{k\geq 0} N_{k+1}(P) T^k
    = -\frac{P'(T)}{P(T)}.
\end{equation}
Also, with the convention that $\sigma_k(P)=0$ for $k>d$ we write
\begin{equation}\label{eq:sersk}
    P(T) = \sum_{k\geq 0} (-1)^k\sigma_k(P)T^k
\end{equation}

Formula \eqref{eq:serhk} is of special interest for Euler
polynomials: the complete homogeneous sums are
the Dirichlet coefficients of primepower index
\[
    \frac{1}{F_p(f,T)}
    = \sum_{e\geq 0} a_{p^e}(f) T^e
    = \sum_{e\geq 0} h_e(F_p(f,T)) T^e.
\]
They obey a finite order recurrence if and only if
the local Euler factor is a rational fraction.

The simplest (and most efficient method for large degree) is
to use formal series computations to convert from one
family to the other.

We compute Newton sums from coefficients
$\sigma_k$ of $h_k$ using an inversion
\begin{equation}\label{eq:HtoN}
    N(P,T) = -\frac{P'(T)}{P(T)} = \frac{H'(P,T)}{H(P,T)}
\end{equation}
and we integrate formally in the other direction
\begin{equation}\label{eq:NtoH}
P(T)^{\pm 1}
= \exp(\pm \log P(T))
= \exp\bigl(\mp\sum_{k\geq 1}N_k(P)\frac{T^k}{k}\bigr).
\end{equation}
This requires the indices $k$ to be invertible in the ring,
which will be the case in our applications (otherwise inversion is
still possible with more coefficients).

We also derive from \eqref{eq:NtoH} a closed formula
for the $k$-th coefficient $h_k(P)$ in terms of Newton power sums.
\begin{prop}\label{prop:hk}
  Let $\Pi_k = \set{ k=\sum_{i=1}^k i m_i , m_i\geq 0 }$ be the set
  of 
  partitions of $k$, then 
  \[
    h_k(P) = \sum_{\Pi_k} \prod_{i=1}^k \frac{ N_i(P)^{m_i} }{ i^{m_i} m_i! }
  \]
\end{prop}
\begin{proof}
We expand
  \[
    \begin{aligned}
        \sum_k h_k(P)T^k
        &=\exp(-\log P(T))
         =\exp\bigl(\sum_{i\geq 1}N_i(P)\frac{T^i}i\bigr)\\
        &=\prod_{i\geq 1}\exp(\frac{N_i(P)}{i}T^i)
         =\prod_{i\geq 1}\sum_{m_i\geq 0} \frac{N_i(P)^{m_i}}{i^{m_i}m_i!}T^{im_i}\\
        &=\sum_{k\geq 0}\bigl(\sum_{\Pi_k} \prod_i\frac{ N_i(P)^{m_i} }{ i^{m_i} m_i!}\bigr)T^k.
    \end{aligned}
  \]
\end{proof}

\begin{ex}
    The first sums are
    \[
      \begin{aligned}
        h_1 &= \sigma_1 = N_1  \\
        h_2 &= \sigma_1^2-\sigma_2 = \frac{N_1^2+N_2}2 \\
        h_3 &= \sigma_1^3-2\sigma_1\sigma_2+\sigma_3
            = \frac{N_1^3+3N_1N_2+2N_3}6 \\
        h_4 &= \sigma_1^4-3\sigma_1^2\sigma_2
               +2\sigma_1\sigma_3+\sigma_2^2-\sigma_4 \\
            &= \frac{N_1^4+6N_1^2N_2+8N_1N_3+3N_2^2+6N_4}{24}
      \end{aligned}
    \]
\end{ex}

We can sum up this section with the following remarks
\begin{itemize}
    \item the coefficients $h_k(P)$ and $N_k(P)$ are solutions
        of the same recurrence equation given by $P(T)$,
        and differ only by initial conditions;
    \item if we know the powers sums $N_k(P)$, the coefficients
        $\sigma_k(P)$ and $h_k(P)$ can be computed
        with almost the same recurrence
        \begin{equation}\label{eq:Nrec}
            \begin{cases}
            k h_k = \sum_{j=1}^k N_jh_{k-j} \\
            k \sigma_k = -\sum_{j=1}^k N_j\sigma_{k-j};
            \end{cases}
        \end{equation}
    \item this recurrence allows to get recursively
        each term $\sigma_k$ or $h_k$ with $k-1$ additions and $k-1$
        multiplications, plus one scalar division by $k$;
    \item conversely $N_k(P)$ can be obtained
        from $\sigma_k(P)$ or $h_k(P)$ at a similar cost, the division
        being replaced by a multiplication.
    \item the simpler recurrence
        \begin{equation}\label{eq:hrec}
            \forall e\geq 1,
            \sum_{i=0}^e (-1)^i a_{p^{e-i}}(f) \sigma_i(F_p(f)) = 0
        \end{equation}
        allows to get $h_k$ recursively
        from coefficients $\sigma_k$ with $k-1$ multiplications
        and $k-1$ multiplications.
\end{itemize}

The point of this section is that representing a polynomial
by its Newton sums $N_k(P)$ instead of its coefficients $\sigma_k(P)$
does not hurt from a computational point of vue.

\begin{lem}\label{lem:costconvert}
    Assume $P(T)$ is a polynomial of degree $d$. For $n\geq 1$, assume
    we know the first $\min(n,d)$ coefficients of one of the series
    $P(T)$, $H(P,T)$, of $N(P,T)$, then we can compute the first
    $n$ coefficients
    of another one (or the same) using $O(n\min(n,d))$ operations.
\end{lem}
\begin{proof}
This is what we obtain if we employ the recurrence formulas
\eqref{eq:HtoN} and \eqref{eq:NtoH}, computing
the degree $d$ polynomial $P(T)$ inbetween if we want more
than $d$ coefficients to ensure that a coefficient of index $k$ is computed
with $\min(d,k)$ operations.
\end{proof}

This is particularly relevant in a context of Artin L-functions:
a representation can be described by its character, whose values
at powers of a conjugacy class are precisely the Newton sums of
the corresponding polynomial.

\begin{rem}
Lemma \ref{lem:costconvert} is not optimal but it is
sufficient in the context of L-functions where these
conversions are a complexity term of secondary order.
\end{rem}

\subsubsection{Operations and Newton sums}

We can compute the operations of Definition \ref{def:poloperations}
thanks to the following relations on Newton sums.
\begin{prop}\label{prop:poloperations}
    With the notations of Definition \ref{def:poloperations},
    we have for all $k,\ell\geq 1$
    \begin{equation}
        \begin{cases}
            N_k(PQ) = N_k(P)+N_k(Q)          \\
            N_k(P\otimes Q) = N_k(P)N_k(Q)          \\
            N_k(P^{\circ \ell}) = N_{k\ell}(P)      \\
            N_\ell(\Sym^k(P)) = h_k(P^{\circ \ell})
        \end{cases}.
    \end{equation}
\end{prop}
\begin{proof}
    The first three equalities are immediate.
    By definition of $\Sym^k(P)$ we have $N_1(\Sym^k(P)) = h_k(P)$.
    This gives the relation for $\ell=1$, then we replace $P$ by
    its $\ell$-th root power $P^{\circ l}$.
\end{proof}

The last relation is used to compute $\Sym^k(P)$ from
its Newton sums. We also consider the more direct expression
obtained from Proposition \ref{prop:hk}, which performs better
for small exponent $k$.
\begin{prop}
    \begin{equation}
        N_\ell(\Sym^k(P))
        = \sum_{\Pi_k} \prod_i \frac{ N_{i\ell}(P)^{m_i} }{ i^{m_i} m_i! }.
    \end{equation}
\end{prop}

\subsubsection{Algorithms}

We record here the algorithms obtained from Proposition \ref{prop:poloperations}.

They rely on conversions between our three families of coefficients that are
already available or easily obtained in computer algebra systems
via the equations \eqref{eq:HtoN} and \eqref{eq:NtoH}.

We write Algorithms \ref{algo:poltensor} and \ref{algo:polsympow}
to have as input and output the usual polynomial coefficients,
but it is straightforward to implement versions taking as input
and returning as output coefficients in any of the three families,
limited to the desired precision, and this does not change the complexity
thanks to Lemma \ref{lem:costconvert}.

\begin{algorithm}[h!]
  \caption{Euler factor of tensor product $\bigotimes_i f_i$}
  \label{algo:poltensor}
   \begin{algorithmic}
      \Input $k$ polynomials $P_i(T) = F_p(f_i,T)$ of degree $d_i$,\\
             precision $\ell$
      \Output first $\ell$ coefficients of the degree $d=\prod_i d_i$
       polynomial $P(T)=F_p(\bigotimes_i f_i,T)$
    \Statex
       \Procedure{SymProduct}{$(P_i)_{1\leq i\leq k}$, $\ell$}
    \State initialize $(N_1,\dots N_{\ell})\gets (1,1,\dots 1)$
    \For{$i\gets 1,\dots k$}
      \State $(M_1\dots M_\ell)\gets $\Call{NewtonSums}{$P_i(T)$, $\ell$}
       \State $(N_1,\dots N_\ell) \gets (N_i\times M_i)_{1\leq i\leq \ell}$ \Comment{componentwise multiplication}
    \EndFor
       \State expand $P(T) = \exp(-\sum_k \frac{N_k}k T^k) \bmod T^{\ell+1}$
       \State \Return $P(T) \bmod T^{\ell+1}$
    \EndProcedure
  \end{algorithmic}
\end{algorithm}

\begin{algorithm}[h!]
  \caption{Euler factor of symmetric power $\Sym^k f$}
  \label{algo:polsympow}
   \begin{algorithmic}
      \Input polynomial $P(T) = F_p(f,T)$ of degree $d$, exponent $k$,\\
       precision $\ell$
      \Output first $\ell$ coefficients of polynomial $Q(T)=F_p(\Sym^k f,T)$
      of degree $\binom{d+k-1}{k}$
      \Statex
      \Procedure{SymPower}{$P$, $k$, $n$}
      \State $(N_1,\dots N_{\ell})\gets (0,\dots 0)$
      \State $(M_1,M_2\dots M_{k\ell})\gets$ \Call{NewtonSums}{$P(T)$, $k\ell$}
      \For{$i\gets 1,\dots \ell$}
        \State $(\tilde M_1,\tilde M_2,\dots \tilde M_k) \gets (M_{i},M_{2i},\dots M_{ki})$
        \State $(\tilde h_1\dots \tilde h_k) \gets $\Call{CompleteSums}{$\tilde M_1,\dots \tilde M_k$}
        \State $N_i \gets \tilde h_k$
      \EndFor
       \State expand $Q(T) = \exp(-\sum_i \frac{N_i}i T^i) \bmod T^{\ell+1}$
       \State \Return $Q(T) \bmod T^{\ell+1}$
      \EndProcedure
  \end{algorithmic}
\end{algorithm}

\begin{prop}\label{prop:costsymprodpow}
    Algorithm \ref{algo:poltensor} computes the first $\ell$ coefficients
    of $\bigotimes_{i=1}^k P_i(T)$ using $O(k\ell^2)$ operations.

    Algorithm \ref{algo:polsympow} computes the first $\ell$ coefficients
    of $P^{\otimes k}(T)$ using $O(k\ell(k+l+d))$ operations.
\end{prop}
\begin{proof}
    In Algorithm \ref{algo:poltensor} we do $(2k+1)$ conversions
    to and from Newton sums, all of which
    are done in $O(\ell^2)$ operations by Lemma \ref{lem:costconvert}.

    In Algorithm \ref{algo:polsympow} we use
    $O(k\ell\min(k\ell,d))$ operations
    to compute the first $k\ell$ Newton sums of $P$, then
    $O(\ell k^2)$ operations to get the first $\ell$ Newton sums of $Q$
    and $\ell^2$ operations to obtain $Q$, hence the $O(k\ell(k+\ell+d))$
    complexity.
\end{proof}

This finishes the proof of Theorem \ref{thm:poloperations}.

\subsubsection{A remark on the Rankin-Selberg relation}\label{sec:rankinselberg}

The Rankin-Selberg method for the product of
two modular forms
$f\in S_{k_1}(N,\phi)$ and $g\in S_{k_2}(N,\psi)$ is written in terms
of the termwise product Dirichlet series
\[
    L(f\odot g,s)=\sum_{n\geq 1}a_n(f)a_n(g)n^{-s}
\]
which is easier to handle than $L(f\otimes g,s)$ for the unfolding
calculations in the upper-half plane.

The calculations make use of the relation
\[
    L(\phi\psi,2s+2-k_1-k_2)\sum_n a_n(f)a_n(g) n^{-s} = L(f\otimes g,s)
\]
where $L(\phi\psi,s)$ is the Dirichlet L-function of the product
of characters, in order to recover the correct product L-function
$L(f\otimes g,s)$.

In order to explain this factor,
let $P(T)=F_p(f,T)$ and $Q(T)=F_p(g,T)$ be two polynomials
of degree at most $2$, then we have an equality
\[
    \sum_{k\geq 0} h_k(P)h_k(Q) T^k = \frac{ 1-\sigma_2(P)\sigma_2(Q) T^2 }{P\otimes Q(T)}
\]

In the case of two modular forms $\sigma_2(P)=p^{k_1-1}\phi(p)$
and $\sigma_2(Q)=p^{k_2-1}\psi(p)$, so the numerator
corresponds to a Dirichlet L-function
\[
    (1-\phi(p)\psi(p)p^{k_1+k_2-2-2s}) = L_p(\phi\psi,2s+2-k_1-k_2)
\]
which explains the nice form of the Rankin-Selberg equation.

For more than two modular forms, or for higher degree polynomials
we still have a relation of the form
\[
    \sum h_k(P)h_k(Q) T^k = \frac{ A(T) }{P\otimes Q(T)}
\]
for some polynomial $A(T)$ of degree less than $\deg(P)\deg(Q)$.

This is a special case of the general fact that sequences
satisfying a finite order linear recurrence are stable by multiplication.
\begin{lem}\label{lem:recprod}
    %
    Let $(u_k)$ and $(v_k)$ be sequences
    satisfying recurrence
    relations given by polynomial $P(T)$ and $Q(T)$.

    Then $(u_kv_k)$ is a recurrent
    sequence for the polynomial $(P\otimes Q)(T)$.
\end{lem}
\begin{proof}
    Let $P(T)=\prod_{i=1}^d (1-\alpha_iT)$. Then $u_k$ is a linear
    combination of the sequences $(\alpha_i^k)_k$ which form a basis
    of the dimension $d$ space of sequences recurrent for $P$.
    Similarly $v_k$ is a sum of the $\beta_j^k$, hence
    $u_kv_k$ is a linear combination of the sequences
    $(\alpha_i\beta_j)^k$ which are recurrent for $P\otimes Q$.
\end{proof}

But there is no reason why the polynomial $A(T)$, which encodes the 
initial conditions describing the particular sequence $h_k(P)h_k(Q)$,
should be expressed in terms of simple L-functions.

For example with two polynomials $P(T)=1-a_1T+a_2T^2-a_3T^3$
and $Q(T)=1-b_1T+b_2T^2-b_3T^3$ of degree $3$ we obtain
\[
    \sum_k h_k(P)h_k(Q)T^k
    =
    \frac{1-A_2T^2+A_3T^3-A_4T^4+A_6T^6}{(P\otimes Q)(T)}
\]
where $A_2=a_2b_2$, $A_3=(a_1a_2-a_3)b_3+(b_1b_3-b_3)a_3$,
$A4=a_1a_3b_1b_3$ and $A_6=(a_3b_3)^2$.

With three polynomials $P,Q,R$ of degree $2$ we obtain a numerator of degree $6$
\[
    1 - BC T^2 + 2ABT^3 - B^2CT^4 + B^3T^6
\]
where $ A = \sigma_1(P)\sigma_1(Q)\sigma_1(R)$,
    $B = \sigma_2(P)\sigma_2(Q)\sigma_2(R)$ and 
    $C = \frac{h_2(P)}{\sigma_2(P)}
       +\frac{h_2(Q)}{\sigma_2(Q)}
       +\frac{h_2(R)}{\sigma_2(R)}$.

\subsection{Euler product expansion}

We advocate that, once a Euler product exists, an efficient
way to represent the Dirichlet series associated
to an arithmetic object $f$ is the family of its Euler factors
$F_p(f,T)$, or even the related data of its Newton powers sums.

Our goal is to prove that we can compute Euler products very
efficiently.
\begin{thm}\label{thm:eulercomplexity}
    Let $R$ be a ring, and $f$ be an arithmetic object whose L-function
    admits a Euler product.
    Let $n>3$, and assume we know
    the Euler factors $F_p(f,T)\in R[T]$
    for all primes $p<n$.

    Then we can compute all coefficients $a_k(f)\in R$ for $k<n$
    with \textbf{less} than $n$ additions and $n$ multiplications in $R$,
    and $O(n)$ operations on indices.

    The conclusion remains if instead of the Euler factors we know
    their Newton sums $N_\ell(F_p(f,T))$ for all $p^\ell<n$,
    and $n\geq 150$.
\end{thm}

Combined with Theorem \ref{thm:poloperations}, we will deduce
the following slightly more general form of
Theorem \ref{thm:prodsymcost}.
\begin{thm}\label{thm:anoperations}
    Let $f_1,\dots f_k$ be $k$ arithmetic objects whose L-functions
    $L(f_i,s)=\sum_{n\geq 1}a_n(f_i) n^{-s}$ admit a Euler product.
    Then from the knowledge of the first $n$ coefficients $a_n(f_i)$
    of each of the $f_i$
    we can compute the first $n$ coefficients
    of $L(\bigoplus f_i,s)$ and $L(\bigotimes f_i,s)$
    with $O(n)$ operations.
    From the knowledge of the first $kn$ coefficients of $f_1$
    we compute the first $n$ coefficients of
    $L(\Sym^k f_1,s)$ with $O(k^2n)$ operations.
\end{thm}

\subsubsection{First algorithm}

Dirichlet coefficients are multiplicative, and
the basic strategy is to propagate primepower values using
the following
principle: for each prime $p$
\begin{enumerate}
  \item 
  we expand the inverse of the Euler factor $F_p(f,T)$ to get
  the coefficients $a_{p^e}(f)$ for exponents $e\geq 1$.
  \item
  for all multiples of $p$ of the form
  $p^em$ with $m$ prime to $p$, we set $a_{p^em}(f) = a_{p^e}(f)a_m(f)$.
\end{enumerate}

The part that can be optimized is the second step.
For comparison purposes we consider default
Algorithm \ref{algo:prodeuler} based on Eratosthene's sieve
to generate composite values $k=p^em$ (instead of factoring numbers).

\begin{algorithm}
    \caption{Expand Euler product $\prod_p F_p(p^{-s})^{-1} \to \sum_n a_n n^{-s}$}
    \label{algo:prodeuler}
  \begin{algorithmic}
      \Input Polynomials $F_p(x)\in 1+xR[x]$ indexed by primes $p<n$
      \Output coefficients $a_1,\dots a_n\in R$
      \Statex
      \Procedure{EulerProduct1}{$n$,$(F_p(x))_{p\leq n}$}
      \State $a_1\gets 1$, $a_2,\dots a_n\gets 0$.
      \For{primes $p<n$}
        \State $d\gets \floor{\log(n)/\log(p)}$
        \State Expand $F_p(x)^{-1}=1+\sum_{e=1}^d c_e x^e + O(x^{d+1})$
        \For{ $e\gets 1$ to $d$ }
          \State $a_{p^e}\gets c_e$
          \State $m \gets 1$
          \While{ $p^e m< n$ }
            \For{ $r$ from 1 to $p-1$ } \Comment{this loop ensures $p\nmid m$}
              \State $a_{p^em} \gets a_m a_{p^e}$ if $p^em<n$ and $a_m\neq 0$.
              \State $m\gets m+1$
            \EndFor
            \State $m\gets m+1$
          \EndWhile
        \EndFor
      \EndFor
      \State \Return $(a_k)_{k\leq n}$
      \EndProcedure
  \end{algorithmic}
\end{algorithm}

In Algorithm \ref{algo:prodeuler}
we loop through cofactors written $m=pq+r$ with $0<r<p$
to ensure that $m$ is prime to $p$, and perform only additions in
the main loop to have fast iteration.

The procedure does $O(n/p)$ operations at each prime $p<n$ for a total
sieve complexity of $O(n\log n)$ operations.
The main inefficiency in Algorithm \ref{algo:prodeuler} is that
we visit each index $k$ several times, one for each of its prime factors,
and we wait until $p$ is the largest prime divisor of $k=p^em$
(so that $a_m$ and $a_{p^e}$ have already been computed)
to actually set the coefficient $a_k=a_{p^e}a_m$.
\begin{lem}
    Algorithm \ref{algo:prodeuler} computes $n$ coefficients
    of an Euler product with $O(n\log n)$ additions and comparisons
    of machine integers and $O(n)$ operations in $R$.
\end{lem}

We can get rid of this unnecessary $\log(n)$ complexity factor.
A simple idea is to precompute for each non prime power index $k<n$
a factorization $k=uv$ into coprime numbers $u,v>1$,
and use it to compute Euler expansions with $O(n)$ operations.

\begin{algorithm}
  \caption{Euler product $\prod F_p(p^{-s}) \to \sum a_n n^{-s}$ with precomputed indices}
    \label{algo:euler-precomp}
  \begin{algorithmic}
      \Input Polynomials $F_p(x)\in 1+xR[x]$ indexed by primes $p<n$ \\
             $P$ set of all prime numbers $p\leq n$\\
             $K$ set of factorizations $k=uv$ into coprime $u,v>0$,\\
             \quad indexed by increasing values of non-primepower $k<n$.
      \Output coefficients $a_1,\dots a_n\in R$
      \Statex
      \Procedure{EulerPrecomp}{$n,P,K,(F_p(x))_{p}$}
      \State $a_1=1$
      \For{primes $p\gets P$}
        \State $d\gets \floor{\log(n)/\log(p)}$
        \State Expand $F_p(x)^{-1}=1+\sum_{e=1}^d c_e x^e + O(x^{d+1})$
        \State $(a_p,a_{p^2},\dots a_{p^d})\gets (c_1,c_2,\dots c_d)$
      \EndFor
      \For{composite $k=p^e\times m \gets K$}
        \State $a_{k}\gets a_{p^e}a_m$
      \EndFor
      \State \Return $(a_k)_{1\leq k\leq n}$
    \EndProcedure
  \end{algorithmic}
\end{algorithm}

\subsubsection{Coprime factorizations}

We now describe how to obtain a list of coprime factorizations.
Algorithm \ref{algo:prodeuler} makes it possible to produce such a list 
using $O(n\log n)$ operations: this is sufficient for practical
purposes.

But in order to obtain a clean complexity statement for Euler product
expansion we reduce the complexity to $O(n)$. The strategy we propose
is based on the following observation.
\begin{itemize}
    \item Algorithm \ref{algo:prodeuler} produces
        factorizations of the form $k=p^em$ with $p$
        the \emph{largest} prime factor of $k$.
        We call it the \emph{smooth} sieve, since
        $m$ is chosen to be $p$-smooth.
    \item the unnecessary operations are concentrated on
        small primes $p$ when there are few $p$-smooth numbers.
        However improving on this is not
        enough since processing all primes $p>\sqrt{n}$
        (for which cofactors $m<n/p$ are automatically $p$-smooth)
        already contributes $O(n\log n)$ operations.
    \item 
      we can improve things by trying to maintain an efficient
      structure of $p$-smooth integers to use for each prime $p$.
      We can obtain good timings in practice but it is not easy to
        get rid of the $\log(n)$ complexity factor: in fact the
        set of smooth integers is a growing set, and there is no
        data structure allowing fast insertion and fast ordered iteration.
\end{itemize}

Hence we reverse the perspective and decide to spend
more time on small primes.

\begin{defn}\label{def:rough}
Let $k>1$ be an integer. We call a factorization
$k = p^e\times m$ where $p$ is prime and $p\nmid m$
\begin{itemize}
  \item
  the \emph{smooth}-coprime decomposition of $k$
  if $m$ is $p$-smooth, i.e. its prime divisors
  are greater than $p$
\item the \emph{rough}-coprime decomposition of $k$
  if $m$ is $p$-rough, i.e. its prime divisors are
  smaller then $p$
\end{itemize}
\end{defn}

If we focus on \emph{rough}-coprime decompositions
we need to update the list of \emph{rough} numbers: this is a
decreasing set, which is easy to keep sorted via a doubly linked list.

This is what Algorithm \ref{algo:coprime} does: we initialize
the list of $2$-rough (i.e. odd) numbers $m\leq n/2$.
\[\begin{tikzcd}[column sep=small,row sep=small]
    1 & 3 & 5 & 7 & 9 & 11 & 13 & 15 & 17 & \dots
	\arrow[shift left, from=1-1, to=1-2]
	\arrow[shift left, from=1-2, to=1-1]
	\arrow[shift left, from=1-2, to=1-3]
	\arrow[shift left, from=1-3, to=1-2]
	\arrow[shift left, from=1-3, to=1-4]
	\arrow[shift left, from=1-4, to=1-3]
	\arrow[shift left, from=1-4, to=1-5]
	\arrow[shift left, from=1-5, to=1-4]
	\arrow[shift left, from=1-5, to=1-6]
	\arrow[shift left, from=1-6, to=1-5]
	\arrow[shift left, from=1-6, to=1-7]
	\arrow[shift left, from=1-7, to=1-6]
	\arrow[shift left, from=1-7, to=1-8]
	\arrow[shift left, from=1-8, to=1-7]
	\arrow[shift left, from=1-8, to=1-9]
	\arrow[shift left, from=1-9, to=1-8]
	\arrow[shift left, from=1-9, to=1-10]
	\arrow[shift left, from=1-10, to=1-9]
\end{tikzcd}\]
They are used to define all rough decompositions $k=2^em$.
Then $1$ links to next prime $p=3$, and we unlink
all multiples $k=3^em$ as soon as they are produced
so that $m$ is always prime to $p$. At the end of the loop
for $p=3$ the list links $3$-rough numbers (all numbers $\pm1\bmod 6$).
\[\begin{tikzcd}[column sep=small,row sep=small]
      & 3 &   &   & 9 &    &    & 15 &    & \dots\\
    1 &   & 5 & 7 &   & 11 & 13 &    & 17 & \dots
	\arrow[from=1-2, to=2-1]
	\arrow[from=1-2, to=2-3]
	\arrow[from=1-5, to=2-4]
	\arrow[from=1-5, to=2-6]
	\arrow[from=1-8, to=2-7]
	\arrow[from=1-8, to=2-9]
	\arrow[shift left, from=2-1, to=2-3]
	\arrow[shift left, from=2-3, to=2-1]
	\arrow[shift left, from=2-3, to=2-4]
	\arrow[shift left, from=2-4, to=2-3]
	\arrow[shift left, from=2-4, to=2-6]
	\arrow[shift left, from=2-6, to=2-4]
	\arrow[shift left, from=2-6, to=2-7]
	\arrow[shift left, from=2-7, to=2-6]
	\arrow[shift left, from=2-7, to=2-9]
	\arrow[shift left, from=2-9, to=2-7]
\end{tikzcd}\]
Now $1$ links to $p=5$ (the first $3$-rough number)
and we go on producing rough decompositions and unlinking the
non $p$-rough numbers.
\[\begin{tikzcd}[column sep=small,row sep=small]
      & 3 &   &   & 9 &    &    & 15 &    & \dots\\
      &   & 5 &   &   &    &    &    &    & \dots\\
    1 &   &   & 7 &   & 11 & 13 &    & 17 & \dots
	\arrow[from=1-2, to=3-1]
	\arrow[from=1-2, to=2-3]
	\arrow[from=1-5, to=3-4]
	\arrow[from=1-5, to=3-6]
	\arrow[from=1-8, to=3-7]
	\arrow[from=1-8, to=3-9]
	\arrow[from=2-3, to=3-1]
	\arrow[from=2-3, to=3-4]
	\arrow[shift left, from=3-1, to=3-4]
	\arrow[shift left, from=3-4, to=3-1]
	\arrow[shift left, from=3-4, to=3-6]
	\arrow[shift left, from=3-6, to=3-4]
	\arrow[shift left, from=3-6, to=3-7]
	\arrow[shift left, from=3-7, to=3-6]
	\arrow[shift left, from=3-7, to=3-9]
	\arrow[shift left, from=3-9, to=3-7]
\end{tikzcd}\]
\begin{algorithm}
  \caption{Sieve for rough-coprime decompositions}
    \label{algo:coprime}
  \begin{algorithmic}
    \Input length $n$
    \Output
      \State list $P$ of all primes $p<n$
      \State list $K$ of non-trivial \emph{rough}-coprime
             decompositions $k = p^em$ for $k<n$
    \Statex
    \Procedure{rough-coprime}{$n$}
    \State $K\gets \varnothing$
    \State $P\gets \varnothing$
    \State $R=\set{1,\dots n}$ \Comment{ linked list of rough numbers }
    \For { $k\gets 1, 3, \dots n$ } \Comment{ Initialize linked list on odd integers $\leq n$ }
       \State set links $k_+\gets k+2$ and $k_-\gets k-2$
    \EndFor
    \While {link $1_+ < n$} \Comment{ Generate }
      \State $p\gets 1_+$ \Comment{first rough is a prime}
      \State store $P \gets P \cup \set{p^e}$
      \For{ powers $p^e<n$ } \Comment{ unlink all $p^e$ }
        \State $(p^e_-)_+\gets p^e_+$
        \State $(p^e_+)_-\gets p^e_-$
      \EndFor
      \For{ powers $p^e<n$ }
        \State $m\gets p_+$ \Comment{ smallest $p$-rough }
        \While{ $m_+ \leq n/p^e$}
          \State $k\gets p^em$
          \State store $K \gets K \cup \set{ k=p^e\times m}$
          \State $(k_+)_-\gets k_-$ \Comment{ unlink $k$ }
          \State $(k_-)_+\gets k_+$
        \EndWhile
      \EndFor
    \EndWhile
    \State \Return $P,K$
    \EndProcedure
  \end{algorithmic}
\end{algorithm}

Since we unlink an element at each step and unlinked elements
are no longer considered the algorithm does $O(n)$ operations.
\begin{prop}\label{prop:coprime}
    Algorithm \ref{algo:coprime} computes all prime numbers
    $p<n$ and the \emph{rough}-coprime
    decomposition of all non primepower $k<n$ using $O(n)$ operations.
\end{prop}

\subsubsection{Proof of Theorems \ref{thm:eulercomplexity} and \ref{thm:anoperations}}

Let $\pi(x)$ denote the prime number counting function;
the number of primepowers less
than $n$ is $\sum_{e\geq 1}\pi(n^{\frac 1e})$.

If we do $e$ arithmetic operations to compute
a coefficient $a_{p^e}(f)$, we obtain a total of
$\sum_{e\geq 1} e\pi(n^{\frac1e})$
operations for all primepower indices.

The following lemma shows that this complexity is always
absorbed by the main $O(n)$ term.

\begin{lem}\label{lem:boundpi}
    For all $n\geq 1$ and all exponent $r\geq0$ we have
    \[
        \sum_{e\geq 1} e^r \pi(n^{\frac1e}) = O(n).
    \]
    We also have the refined inequalities
    \[
        \begin{cases}
        \forall n \geq 2,
                  \sum_{e\geq 2} (e-2)\pi(n^{\frac1e}) < \pi(n) \\
        \forall n \geq 137, 
                  \sum_{e\geq 2} (e-1)\pi(n^{\frac1e}) < \pi(n) \\
        \end{cases}.
    \]
\end{lem}
\begin{proof}
    Since $\pi(n^{\frac1e}=0$ for
    $e>E=\log_2(n)=\frac{\log(n)}{\log(2)}$,
    we take a rough upper bound
    \[
     \begin{aligned}
      \sum_{e\geq 1} e^r\pi(n^{\frac1e})
         &\leq \pi(n)+\sum_{e=2}^E e^r \pi(n^{\frac1e})\\
         &\leq \pi(n) + E^{r+1}\pi(\sqrt{n})
         = O(n) + O(\log(n)^{r+1}\sqrt{n}) = O(n)
      \end{aligned}
    \]

    We refine a little bit for the explicit inequalities:
    we first check numerically
    that they hold on the indicated range for small values $n<3000$.

    For the first inequality we write
    \[
      \sum_{e\geq 3} (e-2)\pi(n^{\frac1e})
      \leq \sum_{e=3}^E (e-2) \pi(n^{\frac13})
      \leq \frac{\log_2(n/2)\log_2(n/4)}2\pi(n^{\frac13})
    \]
    Combined with the absolute bounds \cite[3.5]{RosserSchoenfeld}
    \[
        \begin{cases}
          \forall x>1,  \pi(x) \leq 1.26\frac{x}{\log x}\\
            \forall x>17, \pi(x) \geq \frac{x}{\log(x)}
        \end{cases}
    \]
    we obtain the result since the ratio
    \[
        f(x) = 1.26 \frac{\log_2(x/2)\log_2(x/4) x^{\frac13}\log(x)}{2\log(x^{\frac13})x}
             = 1.26 (\frac32\log_2(x/2)\log_2(x/4)x^{-\frac23})
    \]
    is decreasing and $<1$ for $x=3000$.

    We adapt the argument for the second inequality, counting apart
    the first term $\pi(\sqrt{n})$.

\end{proof}

\begin{prop}
  Combined with Algorithm \ref{algo:coprime},
  Algorithm \ref{algo:euler-precomp} achieves
  the complexity of Theorem \ref{thm:eulercomplexity}.
\end{prop}


\begin{proof}
In Algorithm \ref{algo:euler-precomp}
the propagation step to composite numbers makes exactly
one product in $R$ by non powerprime index,
hence $n-\sum_{e\geq 1}\pi(n^{\frac1e})$ products.

If we assume given the Euler factors $F_p(f,T)$ for $p\leq n$,
we have all for free the coefficients $a_p(f)=\sigma_1(F_p(f,T))$,
and we compute recursively the coefficients
$a_{p^e}(f)$ using \eqref{eq:hrec}
with $e-1$ multiplications and $e-1$ additions in
the coefficient ring $R$, resulting in a total
of $\sum_{e\geq 2} (e-1)\pi(n^{\frac1e})$ operations
for primepower coefficients.

We obtain less than $n$ products (and even less
additions) as soon as
$\pi(n) > \sum_{e\geq 3}(e-2)\pi(n^{\frac1e})$.
This inequality holds by Lemma \ref{lem:boundpi}.

When we replace the data of Euler polynomials
$F_p(f,T)$ by their Newton sums $N_e(F_p(f,T))$,
we use the recursion \eqref{eq:Nrec} instead
and do one extra (small scalar)
division by $e$ to obtain each $a_{p^e}$. This affects
only powers $e\geq 2$, and the conclusion remains thanks
to the last inequality of Lemma \ref{lem:boundpi}.
\end{proof}

\begin{prop}
  Algorithms \ref{algo:euler-precomp} with the Euler factors
  obtained from Algorithms \ref{algo:poltensor} and \ref{algo:polsympow}
  achieves the complexity of Theorem \ref{thm:poloperations}.
\end{prop}
\begin{proof}
Proposition \ref{prop:costsymprodpow} suffices to prove that
    we do $O(\ell^2k)$ operations to compute the first $\ell$
    coefficients of a usual product or symmetric product,
    and $O(\ell^2k^2)$ operations for a symmetric power
    (in fact the complexity is only $O(\ell)$ to compute each
    coefficient of index $\ell$ if computed recursiverly).

    Lemma \ref{lem:boundpi} shows the $O(\ell^2)$ complexity
    in the primepower degrees is absorbed in the $O(n)$ term.
\end{proof}

\subsection{Eisenstein series}

\begin{defn}\label{defn:eis}
    Let $\phi,\psi$ be two primitive characters modulo $N_1$ and $N_2$.
    For $k\geq 1$ we consider the Eisenstein series
    \[
        E_k^{\phi,\psi} = e_k^{\phi,\psi}
        + \sum_{n\geq 1} \bigl(\sum_{d\mid n}\phi(n/d)\psi(d)d^{k-1}\bigr) q^n
    \]
    where the constant term is a value of Dirichlet L-function
    \begin{equation}
        \label{eq:e0}
        e_k^{\phi,\psi}=
        \begin{cases}
            -\frac12 L(\psi,1-k) \text{ if } N_1=1\\
            L(\chi,0)\text{ if } N_2=1 \text{ and }k=1\\
            0 \text{ otherwise}.
        \end{cases}
    \end{equation}
\end{defn}

The value of Dirichlet L-functions at negative integers
can be expressed in terms of the generalized Bernoulli numbers of Leopoldt
\[
    L(\chi,1-k) = \frac{B_{k,\chi}}k \text{ where }
    \sum_k B_{k,\chi}\frac{t^k}{k!}
    = \frac{\sum_{a=1}^N \chi(a)te^{at}}{e^{Nt}-1}.
\]

Following the principles of Section \ref{sec:4products},
we interpret the coefficients for $n\geq 1$ as a
Dirichlet convolution
\[
    a_n(E_k^{\phi,\psi})
    = \sum_{d\mid n}\phi(n/d)\psi(d)d^{k-1}
    = a_n(\phi\oplus \psi\mid_{k-1}).
\]
where $a_n(\phi)=\phi(n)$ and
$ a_n(\psi\mid_{1-k})=\psi(n)n^{k-1}$
are the coefficients of the shift
$L(\psi,s-(k-1))=\sum_n \psi(n)n^{-s+k-1}$.

Hence we obtain the coefficients of $E_k^{\phi,\psi}$
as a Euler product
\begin{equation}\label{eq:aneisenstein}
    \sum_{n\geq 1} a_n(E_k^{\phi,\psi}) n^{-s}
  =
    L(\phi,s)L(\psi,s+1-k)
  =
  \prod_p \frac1{(1-\phi(p)p^{-s})(1-\psi(p)p^{k-1-s})}.
\end{equation}

We detail the procedure in Algorithm \ref{algo:mfeiscoefs},
where we use the fact that the complete sums associated to
$P(T)=(1-\phi(p)T)(1-p^{k-1}\psi(p)T)$ satisfy the recurrence
\[
    h_{e}(P) = (p^{k-1}\psi(p)) h_{e-1}(P) + \phi(p)^e
\]
and $\phi(p)^e = \phi(p^e\bmod N)$ to obtain the number
of multiplications and additions
advertised in Theorem \ref{thm:eiscost}.

\begin{algorithm}
  \label{algo:mfeiscoefs}
  \caption{Eisenstein coefficients via Euler product}
  \begin{algorithmic}
    \Input characters $\phi,\psi$ of modulus $N$, weight $k$, length $n$
    \Output coefficients $a_1(E_k^{\phi,\psi}),\dots a_n(E_k^{\phi,psi})$
    \Statex
    \Procedure{Eisenstein}{$k,\phi,\psi,n$}
    \State compute $v^\phi = (\phi(0),\phi(1),\dots \phi(N-1))\in R$
    \State compute $v^\psi = (\psi(0),\psi(1),\dots \psi(N-1)) \in R$
    \State $P,K\gets$ \Call{rough-coprime}{$n$} \Comment{precompute indices}
    \State $a_1 = 1$
    \For{primes $p\gets P$}
      \State $c_p \gets p^{k-1}v^\psi(p\bmod N_\psi)$
      \State $a_{p} \gets c_p + v^{\phi}(p\bmod N_\phi)$
      \For{ primepower $p^e < n$ }
        \State $a_{p^e} \gets a_{p^{e-1}} c_p + v^{\phi}(p^e\bmod N_\phi)$
      \EndFor
    \EndFor
    \For{composite $k=p^e\times m \gets K$}
      \State $a_{k}\gets a_{p^e}a_m$
    \EndFor
      \State \Return $(a_k)_{1\leq k\leq n}$
    \EndProcedure
  \end{algorithmic}
\end{algorithm}

Algorithm \ref{algo:mfeiscoefs} proves Theorem \ref{thm:eiscost}, since:
\begin{itemize}
    \item we do less than $o_\phi+o_\psi$ multiplications to
        compute the character images and $O(N_\phi+N_\psi)$
        operations to initialize the
        tables of Dirichlet values
    \item we do $O(n)$ operations to compute our
        powerprime and composite tables using Algorithm \ref{algo:coprime}.
    \item we do exactly one multiplication and at most one addition
        per coefficient $a_k$.
\end{itemize}
When $k=1$, we even skip the $\pi(n)$ multiplications by $p^{k-1}$.

\section{Coefficients of modular forms}

Let $k\geq 2$ and $f\in S_k(\Gamma_1(N))$ be an eigenform,
we consider the problem
of computing its Fourier coefficients at the cusp $\infty$
\[
    f(z) = \sum_{n\geq 1}a_n(f) q^n.
\]

\subsection{Borisov-Gunnels expression}

\newcommand\B{{\vert B}}

We use the notations of \cite{DicksonNeururer}: for $f\in M_k(\Gamma_1(N))$,
in particular the expansion operator $f\B_d(z) = d^{\frac k2}f(dz)$.

We assume given an explicit decomposition in terms of Eisenstein series
\[
    f =
    \sum_{i=1}^r c_i
    \bigl(E_{\ell_i}^{\chi_i,\psi_i}\B_{d_1}\bigr)
    \bigl(E_{k-\ell_i}^{\chi_i,\psi_i} \B_{d_2}\bigr)
\]
where by abuse of notation we allow terms involving $E_0=1$:
they correspond to elements in the Eisenstein subspace
$\mathcal E_k\subset M_k(N)$.

We do not enter the problem of finding such a decomposition:
there is a rich literature on the theoretical side asserting that
certain families are sufficient to generate $M_k(\Gamma_1(N))$
if $k>2$. For instance we quote the following
from \cite[theorem 5.4]{DicksonNeururer}.

\begin{thm}\label{thm:borisov-gunnels}
  Let $k\geq 4$ and assume that $N$ divided by its squarefree part has at most two prime factors.
  Then the vector space $M_k(\Gamma_0(N))$ is generated by
  its Eisenstein subspace $\mathcal E_k(N)$ and the products
  \[
      E_\ell^{\phi,\psi}\mid B_{d_1 d}\cdot
      E_{k-\ell}^{\phi^{-1},\psi^{-1}}\mid B_{d_2d}
  \]
  such that
  $1\leq \ell\leq k-1$,
  $\phi$ and $\psi$ are primitive characters modulo $N_1$
          and $N_2$ with $\phi\psi(-1)=(-1)^\ell$,
  $d_1M_1$ divides the squarefree part of $N$ and
  $N_1d_1N_2d_2d\mid N$.
\end{thm}

A well understood caveat is that in weight $k=2$ the space
generated by products of Eisenstein series avoids the
subspace of modular forms $f$ such that $L(f,1)=0$.
In this case we can
use a division trick employed by Belabas and Cohen: take
any weight 1 modular form $g\in S_1(\Gamma_1(N))$
and express $fg\in S_3(\Gamma_1(N))$ using Eisenstein series
(since the decomposition always exists in weight $k=3$).
This allows fast computation of coefficients of $f$ at the
cost of one extra division of power series.

In practice, when the ambiant space $S_k(N)$ is available
it suffices to introduce products of Eisenstein series
and check by linear algebra that our modular form belongs
to their span.

This functionality is already present in the source code of
Pari/GP \cite{PariGP} and used internally for computing expansions at
cusps different from $\infty$ \cite{BelabasCohen}.

\begin{ex}\label{ex:bg}
    We give here a sample list of decompositions. For Dirichlet
    characters the notation $\chi_N(a,\cdot)$ refers to the explicit
    isomorphism $(\Z/N\Z)^\times\to\widehat{(\Z/N\Z)^\times}$ introduced
    by Conrey and described in \cite{LMFDB:conrey}.
    \begin{itemize}
        \item $N=11$, $k=2$, $f_{11}=q-2q^2-q^3\dots $ the newform in $S_2(11)$,
            let $\chi=\chi_{11}(1,\cdot)$ be the trivial character modulo $11$ and
            $\psi=\chi_{11}(-1,\cdot)$ the character of order $2$, then
            \[
                f_{11} = -\frac32 E_2^{1,\chi} + \frac52 \bigl(E_1^{1,\psi}\bigr)^2
            \]
       \item $N=23$, $k=2$, $f_{23}= q + (y-1)q^2 + (1-2y)q^3 - yq^4\dots$
            the newform in $S_2(23)$ having coefficients in $\Z[y]$ where $y^2=y+1$.
            Let $\chi=\chi_{23}(1,\cdot)$ be the trivial character modulo $23$,
            $\psi=\chi_{23}(5,\cdot)$ the character
            such that $\psi(5)=\zeta$ where $\zeta$ is a $22$-th root of unity,
            and $\phi=\psi^{11}=\chi_{23}(-1,\cdot)$ the quadratic character.
            Then
            \[
                f_{23} = \frac{
                    (27c+3-9y) E_2^{1,\chi}
                    -(11c-y+3) \bigl(E_1^{1,\phi}\bigr)^2
                    +u(6y+4)
                    E_1^{1,\psi}E_1^{1,\psi^{-1}}
                }{4a-6b}
            \]
            where we write
            \[
              \begin{cases}
                a = (1-\zeta)(1+\zeta^2)(1+\zeta^4)^2 \\
                b = \zeta^3(1-\zeta)(1+\zeta^4) \\
                c = (a-1)y + b-1 \\
                u = \zeta^9 - \zeta^6 + \zeta^5 - \zeta^2 + 2 \\
              \end{cases}.
            \]
        \item $N=32$, $k=2$,
            $f_{32}=q-2q^5-3q^9+6q^{13}\dots$ the newform in $S_2(32)$.
            Let $\zeta$ a height root of unity and $\psi=\chi_{23}(3,\cdot)$
            the odd character modulo $32$ such that $\psi(5)=\zeta^3$, then
            \[
                f_{32} = (\zeta^3-\zeta+1)
                    (-E_2+E_2\B_2-8E_2\B_{16}+32E_2\B_{32})
                    +(\zeta^3-\zeta) E_1^{1,\psi} E_1^{1,\overline \psi}
            \]
            Here the cuspidal space has dimension $1$ so we manage to
            have an expression with only one product.
        \item $N=43$, $k=2$,
            $f_{43}= q + yq^2 - yq^3 + (2-y)q^5\dots$ with
            coefficients in $\Z[y]$ where $y^2=2$,
            then 
            \[
                f_{43} =
                \frac{y-1}2 E_2^{1,\chi}
                +\frac{3y+5}6 \bigl(E_1^{1,\phi}\bigr)^2
                +\frac{2-3y}3 E_1^{1,\psi}E_1^{1,\psi^{-1}}
            \]
            where $\chi=\chi_{43}(1,\cdot)$, $\phi=\chi_{43}(-1,\cdot)$,
            $\psi=\chi_{43}(7,\cdot)$.
    \end{itemize}
\end{ex}

One can compare these decompositions to those of \cite{DicksonNeururer}:
we try to minimize the number of
products and favor the use of expansion operators $B_d$
(for which we can recycle the same coefficients).
We manage to find expressions having $\dim S_k(N)$ products
instead of $\dim M_k(N)$: when $N$ is composite this makes a huge
difference.

For computational purposes we adopt the following notation.
\begin{defn}\label{defn:bgdec}
    Let $f\in S_k(\Gamma_1(N))$, a \emph{Borisov-Gunnels}
    decomposition of $f$ is a data $(k,N,\cC,\cE,\Omega,B)$ where
    \begin{itemize}
        \item $\cC$ is a list of characters $(\phi_1,\dots \phi_r)$
            modulo $N$
        \item $\cF$ is a list of $m$ tuples of the form
            $(\ell,i,j,i',j',d,d')$, each encoding a product
            $E_\ell^{\phi_i,\phi_j}\B_{d}E_{k-\ell}^{\phi_{i'},\phi_{j'}}\B_{d'}$
            (one factor being trivial if $\ell=k$)
        \item $\Omega=(\omega_1,\dots \omega_d)$ is the description
            of an integral basis of the
            ring of integers $\Z[f]$
        \item $B$ is a matrix of size $d\times m$ representing
            the coefficients of $f$ on the family $\cF$.
    \end{itemize}
\end{defn}

The example of $f_{23}$ demonstrates that the coefficients of
the decomposition usually involve large cyclotomic integers.
It is better to adopt
a representation modulo a prime $q$, so that the matrix $B$ belongs to
$M_{d\times m}(\F_q)$ for a suitable prime $q$, as we explain now.


\subsection{Modular computation}

We have Ramanujan bound on eigenform coefficients, so that a form can be
recognized, and more importantly its coefficients can be computed
working modulo a prime $q$.

More precisely we convert the Ramanujan bound into
an explicit bound on the reduced coefficients:
\begin{lem}\label{lem:hasse}
    Let $f$ be a weight $k$ eigenform, whose
    coefficients are integers in the Hecke field $E=\Q(f)$.

    Let $\omega_1,\dots \omega_d$ be an integral basis of
    $\CO_E$, let
    \[
        \Omega = \bigl(\omega_i^\sigma)_{i,\sigma}\bigr)
    \]
    be the matrix of all embeddings $\sigma:E\to\C$
    and denote by
    \[
        B = \max_{i,j}\abs{(\Omega^{-1})_{i,j}}
    \]
    the maximum coefficient of the inverse matrix.
    
    Then the coefficients of the decomposition of $a_n(f)$
    on the basis $(\omega)_i$
    \[
        a_n(f) = \sum_{i=1}^d a_{n,i}(f)\omega_i
    \]
    satisfy
    \[
        \abs{ a_{p,i} }\leq 2\sqrt{p}^{k-1}B.
    \]
\end{lem}
\begin{proof}
    This follows from the Hasse-Ramanujan bound
    $\abs{ a_p^\sigma }\leq 2\sqrt{p}^{k-1}$
    for each embedding $\sigma:E\to\C$
    and the relation
    $(a_p(f)^\sigma)_\sigma = (a_{p,i}(f))_i \Omega$
    between embeddings and integral basis coefficients.
\end{proof}

\subsection{FFT product}

Now let $q$ be a prime chosen
such that all characters $\phi,\psi$
have values in $\F_q^\times$
(a sufficient condition  being $q\equiv 1\bmod \varphi(N)$).

Assume also that $q$ is a FFT prime, meaning $\F_q^\times$
contains a root of unity of order $2^r$ with $2^r$ at least
twice the number $n$ of coefficients we want to compute.

Then the first $n$ coefficients of the product of powers series 
$E_\ell^{\phi,\psi}E_{\ell'}^{\phi',\psi'}$ can be computed modulo
$q$ in time $O(r2^r)=O(n\log n)$.

If $q$ is large enough so that conditions of Lemma \ref{lem:hasse}
are satisfied, we can lift the coefficients $a_p(f)$
to the integer ring $\CO_E$.

Otherwise we can replicate the computation modulo other primes $q$
until the result can be lifted, using a CRT reconstruction.

\subsection{Algorithm}

Algorithm \ref{algo:mfcoefs} describes the process. We choose to
discard coefficients of non prime index: this allows to gain
a factor $\log(n)$ on the memory footprint and on the
linear algebra step to combine generators. If all coefficients
are needed we simply apply Algorithm \ref{algo:euler-precomp}
at the end.

\begin{algorithm}
    \label{algo:mfcoefs}
  \caption{Modular form coefficients}
  \begin{algorithmic}
    \Input Borisov-Gunnels decomposition $f=(k,N,\cC,\cE,\Omega,B)$, length $n$
    \Output coefficients $a_1(f),\dots a_n(f)$ on basis $\Omega$.
    \Statex
    \Procedure{ModularFormCoefficients}{$f,n$}
    \State $q\gets $ characteristic of $B$
    \Comment{All computations are done modulo $q$}
    \State $o\gets \lcm\set{\order(\chi_i), \chi_i\in\cC}$
      \For{$\chi\in\cC$}
      \State $v^\chi\gets (\chi(0),\dots \chi(N))$ \Comment{Precompute values}
      \EndFor
      \State $P,K\gets $\Call{rough-coprime}{$n$} \Comment{Precompute indices}
    \State $M \gets 0$, matrix of size $\#\Omega\times \#P$.
    \For{$(\ell,i,j,i',j',d,d')\gets \cE$}
      \State $(a_n)\gets $\Call{Eisenstein}{$\ell,\phi_i,\phi_j,n$}
      \State $(b_n)\gets $\Call{Eisenstein}{$k-\ell,\phi_{i'},\phi_{j'},n$}
      \State $(c_n)\gets$ \Call{FFT-multiplication}{$a_n,b_n$ mod $q$}
      \State $M\gets M\cup (c_p)_{p\in P}$
    \EndFor
    \State $(a_p)_p \gets B \times M$ modulo $q$
      \State lift $a_p$ in $\Z$
      \State \Return $(a_p)$
    \EndProcedure
  \end{algorithmic}
\end{algorithm}

\begin{cor}\label{cor:hasse}
    Let $f\in S_2(N)$ be a modular form given by a Borisov-Gunnels
    decomposition, and let $B=B(f)$ the bound of
    Lemma \ref{lem:hasse} associated to the integer basis.

    Let $\ell = 2^rm+1$ be a prime number, and assume
    that $B\leq \frac{\ell-1}4$. Then for all length
    $n\leq 2^{r-1}$ such that $\sqrt n B \leq \frac{\ell-1}4$,
    Algorithm \ref{algo:mfcoefs} computes
    all values $a_p(f)\in\CO_E$ for $p\leq n$ in time $O(n\log n)$.
\end{cor}    
\begin{proof}
    A binary Fast Fourier Transform modulo $\ell$ allows to compute
    up to $2^{r-1}$ coefficients, and the reductions $a_{p,i}(f)$
    can be lifted to $\Z$ if $2\sqrt{p}B\leq\frac{\ell -1}2$.
\end{proof}

Due to memory issues we have limited our implementation
to lengths $n < 2^{32}$, so that for small levels
$N<1000$ it is easy to find FFT primes of 53 bits
such that $q\equiv 1\bmod \varphi(N)$
and $\equiv 1\bmod 2^{33}$.


\section{Timings}\label{sec:timings}

We reproduce in Table \ref{tab:timings} the timings obtained by our
implementation for a few example forms.
Forms of level $N=11$ and $N=23$
are listed on page \pageref{ex:bg}. Forms of level $N=41$ and $N=131$
have coefficients in fields of degree $d=3$ and $d=10$.

We also computed the expansions of the two eigenforms
of weight $k=4$ and level $N=13$, one of which is rational and the other has
coefficients in $\Q(\sqrt{17})$. Both expansions are obtained at the same time
since the forms are expressed on the same basis of Eisenstein products.

We detail the portion of time spent on computing all Eisenstein series $E_i$
and on the FFT products $E_i\cdot E_i'$.

The timings for lengths $n=10^5$ and $10^7$ are obtained as single core
computation on a laptop equipped with a M4 processor. This machine has
not enough memory to run the computation for length $n=10^9$, the
corresponding timings are obtained on a AMD Epyc server with $200GB$
of memory

Note that the algorithm is ``embarassingly parallel'' with respect to
the computation of products of Eisenstein series, which can be
processed simultaneously.

This improvement is useful when computing modular forms of larger
level. We did not implement it yet, being more concerned with
memory issues than with processing time.

\begin{table}
    \centering
    \begin{tabular}{ccccccc}
\toprule
$N$         & $k$    & $d=[\Q(f):\Q]$ & length     & Eisenstein & FFT     & total \\
\midrule
11          & 2      & 1              & $10^5$     & 3ms        & 5ms     & 12ms \\
            &        &                & $10^7$     & 94ms       & 286ms   & 507ms \\
            &        &                & $10^9$     & 36.5s      & 1mn16s  & 2mn30s \\
\midrule \\
13          & 4      & 1+2            & $10^5$     & 1ms        & 7ms     & 10ms \\
            &        &                & $10^7$     & 174ms      & 989ms   & 1275ms \\
            &        &                & $10^9$     & 2mn1s      & 5mn16s  & 8mn \\
\midrule \\
23          & 2      & 2              & $10^5$     & <1ms       & 4ms     & 5ms \\
            &        &                & $10^7$     & 109ms      & 502ms   & 711ms\\
            &        &                & $10^9$     & 55.9s      & 2mn20s  & 3mn55s \\
\midrule \\
41          & 2      & 3              & $10^5$     & 1ms        & 6ms     & 8ms\\
            &        &                & $10^7$     & 147ms      & 738ms   & 998ms\\
            &        &                & $10^9$     & 1mn59s     & 3mn38s  & 6mn17s \\
\midrule \\
131         & 2      & 10             & $10^5$     & 5ms        & 17ms    & 24ms \\
            &        &                & $10^7$     & 426ms      & 2394ms  & 3024ms \\
            &        &                & $10^9$     & 5mn20s     & 12mn31s & 18mn56s \\
\bottomrule
\end{tabular}
    \caption{Timings for our algorithm, with time spent on
the $2d+1$ Euler products and $d$ FFT multiplications.}
\label{tab:timings}
\end{table}

\subsection{Coefficients of modular forms}\label{sec:compare}

For comparison purposes we give a list of alternative timings using the
default modular form functionality of Sage, Magma and Pari/GP.
None of these packages has been developped
in the perspective of computing large q-expansions.

We consider the simplest example of the level 11
newform $f\in S_2(\Gamma_0(11))$, and detail
below various ways of computing its coefficients.
All timings are collected in Table \ref{tab:compare}.

\begin{table}
    \centering
\begin{tabular}{llccc}
    \toprule
    form                        & length & Pari/GP   & Sage & Magma \\
    \midrule
    $f\in S_2(\Gamma_0(11))$    & $10^4$ & 600ms & 22s & 20s \\
                                & $10^5$ & 18s   & 34mn & 50mn \\
    \midrule
    $f=(\eta(z)\eta(z^{11}))^2$ & $10^5$ & 159ms &  & \\
                                & $10^7$ & 29s   &  & \\
    \midrule
    $L(E_{11},s)$               & $10^5$ & 106ms & 195ms & 650ms \\
                                & $10^7$ & 9.2s   & 24s   & 78s \\
    \bottomrule
\end{tabular}
    \caption{Timings for alternative methods in level $N=11$}
    \label{tab:compare}
\end{table}

\subsubsection{Generic modular forms expansions}

Pari/GP computes the coefficients from its internal
representation based on the trace formula \cite{BelabasCohen}.
\begin{code}
  \label{code:mfcoefspari}
  \caption{modular form coefficients with Pari}
\begin{alltt}
mf = mfinit([11,2],1);
f=mfeigenbasis(mf)[1];
mfcoefs(f,10^4);
time = 601 ms.
mfcoefs(f,10^5);
time = 18,021 ms.
\end{alltt}
\end{code}

While Magma and Sage use modular symbols.
\begin{code}
  \label{code:mfcoefsmagma}
  \caption{modular form coefficients with Magma}
\begin{alltt}
S2_11 := CuspidalSubspace(ModularForms(Gamma0(11),2));
f := S2_11.1;
time an := qExpansion(f, 10^4);
Time: 20.870 seconds
time an := qExpansion(f, 10^5);
Time: 2969.860 seconds
\end{alltt}
\end{code}

\begin{code}
  \label{code:mfcoefssage}
  \caption{modular form coefficients with Sage}
\begin{alltt}
S2_11 = CuspForms(Gamma0(11));
f = S1_11.0
\%time an = f.q_expansion(10^4);
CPU times: user 22.5 s, sys: 252 ms, total: 22.7 s
\%time an = f.q_expansion(10^5);
CPU times: user 33min 49s, sys: 30.1 s, total: 34min 19s
\end{alltt}
\end{code}

\subsubsection{Other representations}

For this particular newform of level 11 there are efficient alternatives.

For example we can use the eta product expansion
\[
    f(z) = \eta(z)\eta(z^{11}) = \prod_n(1-q^n)^2(1-q^{11n})^2
\]
which gives an algorithm whose complexity is dominated by
a number of FFT multiplication.

\begin{code}
  \label{code:etapari}
  \caption{eta product in Pari/GP}
\begin{alltt}
f(n)=(eta('q+O('q^n))*eta('q^11+O('q^n)))^2;
f(10^5);
time = 159ms.
f(10^7);
time = 29,465 ms.
\end{alltt}
\end{code}

Pari offers this possibility but the generic implementation of
eta products could be improved on this example.
\begin{code}
  \label{code:mfetaquo}
  \caption{modular form as eta product with Pari}
\begin{alltt}
f = mffrometaquo([1,2;11,2]);
mfcoefs(f,10^5);
time = 204 ms.
mfcoefs(f,10^7);
time = 40,851 ms.
\end{alltt}
\end{code}

We can also compute these coefficients via the associated elliptic curve.

\begin{code}
  \label{code:ellanpari}
  \caption{L series of elliptic curve with Pari}
\begin{alltt}
ell = ellinit("11a1");
ellan(ell,10^5);
time = 106 ms.
ellan(ell,10^7);
time = 9,266 ms.
\end{alltt}
\end{code}

\begin{code}
  \label{code:ellansage}
  \caption{L series of elliptic curve with Sage}
\begin{alltt}
e = EllipticCurve("11a");
\% time an = E.anlist(10^5);
CPU times: user 191 ms, sys: 4.22 ms, total: 195 ms
\% time an = E.anlist(10^7);
CPU times: user 23.6 s, sys: 734 ms, total: 24.3 s
\end{alltt}
\end{code}

\begin{code}
  \label{code:ellanmagma}
  \caption{L series of elliptic curve with Magma}
\begin{alltt}
e := LSeries(EllipticCurve("11a1"));
time an := LGetCoefficients(e,10^5);
Time: 0.650 seconds
time an := LGetCoefficients(e,10^7);
Time: 78.210 seconds
\end{alltt}
\end{code}

With this small conductor $N$ our method outperforms
default implementations of point counting:
this is not a surprise when traces $a_p$
are obtained with Shanks-Mestre's baby-step giant-step algorithm
(the method is used for $p<2^{40}$ in Pari/GP) which results
in a $O(n^{5/4})$ complexity for computing $n$ coefficients.
We see a clear difference for lengths $n>10^6$.

Of course this example is very specific:
point counting modulo $p$ is almost insensitive
to the conductor $N$ while our method has $O(N)$ complexity.

\subsection{Coefficients of Eisenstein series}

We demonstrate that it is crucial to adopt
a fast algorithm for the generation of coefficients of Eisenstein series,
as we did with Algorithm \ref{algo:mfeiscoefs}: we see that
even with a very simple quadratic character
$\phi=\chi_{23}(-1,\cdot)$,
current implementations compute the coefficients
of $E_1^{1,\phi}$
in a much slower way than the product $(E_1^{1,\phi})^2$.

In this context a Borisov-Gunnels decomposition
leads to much worse complexity than $O(n\log n)$.

\begin{code}
  \label{code:eiscoefspari}
  \caption{Eisenstein series coefficients with Pari}
\begin{alltt}
e = mfeisenstein(1,Mod(-1,23));
an = mfcoefs(e,10^5);
time = 440 ms.
Ser(an)^2;
time = 36ms.
an = mfcoefs(e,10^6);
time = 5,131 ms.
Ser(an)^2;
time = 246 ms.
\end{alltt}
\end{code}

\begin{code}
  \label{code:eiscoefsmagma}
  \caption{Eisenstein series coefficients with Magma}
\begin{alltt}
M := ModularForms(Gamma1(23),1); M;
E := EisensteinSubspace(M); E;
e := EisensteinSeries(E)[1]; e;
3/2 + q + 2*q^2 + 2*q^3 + 3*q^4 + 4*q^6 + 4*q^8 + 3*q^9 + O(q^12)
time an:= qExpansion(e,10^5);
Time: 47.560 seconds.
time bn:= an^2;
Time: 0.020
\end{alltt}
\end{code}

\begin{code}
  \label{code:eiscoefssage}
  \caption{Eisenstein series coefficients with Sage}
\begin{alltt}
M = ModularForms(Gamma1(23),1,eis_only=True)
E = M.eisenstein_series()
e = E[5]
3/2 + q + 2*q^2 + 2*q^3 + 3*q^4 + O(q^6)
\%time an = e.q_expansion(10^5);
Wall time: 10.2 s
\%time bn = an^2;
Wall time: 12.9 ms
\end{alltt}
\end{code}

\section{Application: triple product L-functions}

Gross and Kudla \cite{GrossKudla1992} investigated for three
weight 2 modular forms $f_1,f_2,f_3$ of the same level $N_0$
the product
\[
  L(f_1\otimes f_2\otimes f_3,s)
  =\prod_p F_p(f_1\otimes f_2\otimes f_3,p^{-s})^{-1}
\]
where for all $p\nmid N_0$ the local factor is the tensor
product of Definition \ref{def:poloperations}.

For $p\mid N_0$ the local factor has to be adjusted with
the following formula
\[
    F_p(f_1\otimes f_2\otimes f_3, T)
    = (1-\alpha T)(1-p\alpha T^2)
    \text{ where }
    \alpha = a_p(f_1)a_p(f_2)a_p(f_3)
\]

The weight is $k=2+2+2 - 2 = 4$, the conductor $N=N_0^5$
and the Gamma factors $\gamma(s) = \Gamma_\C(s)\Gamma_\C(s-1)^3$,
and the completed L-function satisfies
\[
    \Lambda(f_1\otimes f_2\otimes f_3,s)
    = N_0^{\frac 52s}\gamma(s)L(f_1\otimes f_2\otimes f_3,s)
    = \overline{\Lambda}(f_1\otimes f_2\otimes f_3,4-s).
\]

Once all these elements are available, we use our
generic L-function facilities that are
part of Pari/GP since 2015 \cite{parigp:lfun}.
The function \texttt{lfuninit} computes a trigonometric polynomial
$P_\Lambda(x)$ which approximates
the completed L-function to the required precision
near a finite segment of the critical line
\[
    \Lambda(f_1\otimes f_2\otimes f_3, 2+it)
    \approx P_\Lambda(e^{it})
    \text{ for }t\in[-T,T]
\]

For small precision and range $T$ it suffices to compute
an approximation
polynomial $P_\Lambda(x)=\sum_{\ell=-L}^L w_\ell x^\ell$
of small degree $L$, but whose
coefficients are very long smoothed Dirichlet sums of the form
\[
    w_\ell = he^{\frac{k\ell h}{2}}
    \sum_{n\geq 1}
    a_n(f_1\otimes f_2\otimes f_3)
    K_\gamma(\frac{ne^{\ell h}}{\sqrt{N}})
\]
where $K_\gamma(t)$ is the inverse Mellin transform of the
Gamma factor $\gamma(s)$, and $h>0$ is a discretization
parameter.

As an example, let us consider the three newforms
of level $N=35$
\begin{equation}
    \begin{cases}
    f = q + q^3 - 2q^4 - q^5 + q^7 - 2q^9 + O(q^{10})\\
    g = q + (y - 1)q^2 - yq^3 + (-y + 3)q^4 + q^5 - 4q^6 - q^7 + (y - 5)q^8+ O(q^9)\\
    h = q - yq^2 + (y - 1)q^3 + (y + 2)q^4 + q^5 - 4q^6 - q^7 + (-y - 4)q^8 + O(q^9)
    \end{cases}
\end{equation}
where $y=\frac{1+\sqrt{17}}2$ is a root of $y^2-y-4$
($g$ and $h$ are Galois conjugate).

With the notations of Definition \ref{defn:bgdec}, they can
be expressed using Eisenstein series in the following way.
We consider the characters
\[
  \begin{aligned}
    \cC &= (\phi_1,\phi_2,\phi_3,\phi_4,\phi_5=\phi_4^{-1})\\
        &= (\chi_1(1,\cdot), \chi_{35}(6,\cdot),
            \chi_{35}(-1,\cdot), \chi_{35}(8,\cdot), \chi_{35}(29,\cdot)),
  \end{aligned}
\]
the basis of forms
\[
\cF = 
\begin{cases}
    (2,1,1,1,1,1,1) = E_2^{1,1} \\
    (2,1,1,1,1,5,1) = E_2^{1,1}\B_5 \\
    (2,1,1,1,1,7,1) = E_2^{1,1}\B_7 \\
    (2,1,1,1,1,35,1) = E_2^{1,1}\B_{35}\\
    (1,1,2,1,2,1,1) = (E_1^{1,\phi_2})^2 \\
    (1,1,3,1,3,1,1) = (E_1^{1,\phi_3})^2 \\
    (1,1,4,1,5,1,1) = E_1^{1,\phi_4}E_1^{1,\phi_4^{-1}}
\end{cases}
\]
and the coefficients of $f$ and $g$ on $\cF$ are
\[
  \begin{cases}
    B_f = ( -\tfrac72, \tfrac{19}2, \tfrac{13}2,
             \tfrac{47}2, \tfrac98, -\tfrac38, \tfrac{15}4 )  \\
    B_g = ( -\tfrac32y+\tfrac75, \tfrac72y-7,
             \tfrac92y-\tfrac195,\tfrac{11}2y+19,
             \tfrac78y-\tfrac12,-\tfrac58y+\tfrac32,\tfrac54y-3)
  \end{cases}
\]

For an absolute precision of $b=25$ bits and height $t<30$ on the
critical line, Pari expects 30 million Dirichlet coefficients.
\begin{code}
  \label{code:lfuntriple}
  \caption{Complexity estimates for triple product L-function}
\begin{alltt}
Lshape = lfuncreate([n->[],1,[-1,-1,-1,0,0,0,0,1],4,35^5,1]);
localbitprec(25);lfuncost(Lshape,[30])
\% = [30848082, 322]
\end{alltt}
\end{code}

We precompute $40$ million coefficients $a_p(f)$ and $a_p(g)$
with our implementation, we conjugate to get the $a_p(h)$
and run Algorithms \ref{algo:poltensor} and \ref{algo:prodeuler}
to obtain the first $40$ million coefficients
\[
    (a_n(f\otimes g\otimes h))_n
    = (1,0,-4,8,-11,0,15,0,13,0,12,-32,10,\dots)
\]
The whole process takes a few seconds.

Now Pari computes all the weights $w_\ell$ and we obtain an interpolation
polynomial for $\Lambda(f\otimes g\otimes h,2+it)$, $t\in[-30,30]$.

\begin{code}
  \label{code:lfunplot}
    \caption{Graph of the Z-function of $L(f\otimes g\otimes h,s)$}
\begin{alltt}
an = [1,0,-4,8,-11,0,15,0,13,0,12,-32,10,0,44,48,6,0,-16,-88,-60...
Lfgh = lfuncreate([an,0,[-1,-1,-1,0,0,0,0,1],4,35^5,1]);
default(bitprecision,20)
init = lfuninit(Lfgh,[30]);
\end{alltt}
\end{code}

This last computation takes one hour: it is clear
that all the computation time is now spent on
the weights $w_\ell$. Many things can be improved
on this matter, this is another issue we shall
consider later to make this kind of computation
more accessible.

Once the approximation polynomial $P_\Lambda$ is computed
we obtain the plot of the Hardy $Z$-function
of $L(f\otimes g\otimes h,s)$ (Figure \ref{fig:lfunplot}) or the first
zeros on the critical line (Table \ref{tab:lfunzeros}).

\begin{figure}
  \centering
  \includegraphics[width=.8\linewidth]{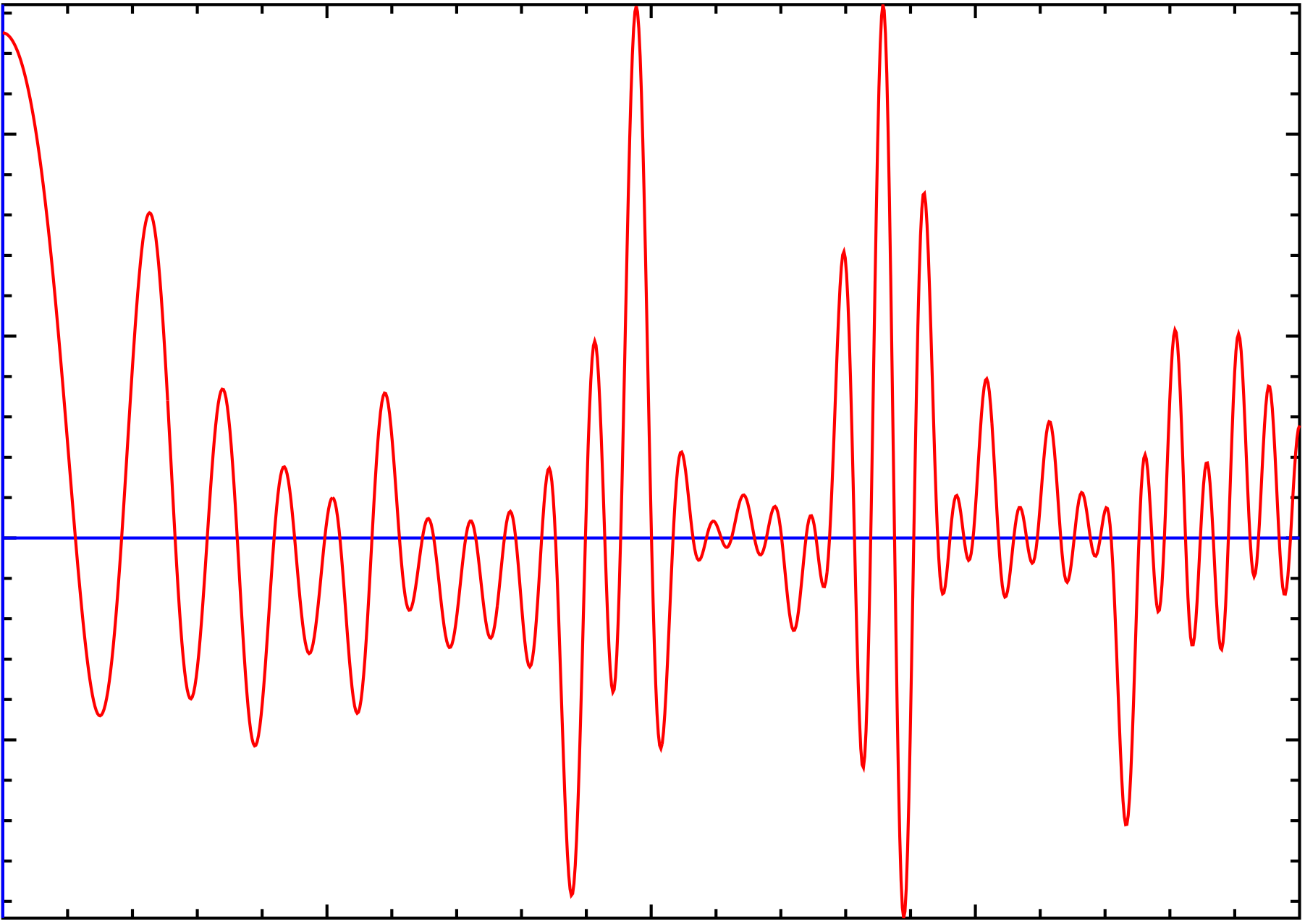}
  \caption{Graph of the $Z$-function of $L(f\otimes g\otimes h,s)$ on $[0,20]$}
  \label{fig:lfunplot}
\end{figure}

\begin{table}
  \centering
  \begin{tabular}{rrrrr}
      \toprule
   1.1202978  &   6.6470076  &   10.645199 & 13.753081 & 16.773863 \\
   1.8330179  &   7.1397898  &   10.851837 & 14.050219 & 16.925990 \\
   2.6575646  &   7.2944945  &   11.080469 & 14.419385 & 17.093806 \\
   3.1533236  &   7.7361099  &   11.239746 & 14.608319 & 17.523444 \\
   3.6159878  &   7.9107288  &   11.592976 & 14.823733 & 17.723815 \\
   4.1842413  &   8.3109855  &   11.777484 & 14.964649 & 17.914793 \\
   4.5009449  &   8.5242661  &   12.012070 & 15.356085 & 18.240817 \\
   4.9662033  &   8.9862883  &   12.387441 & 15.592219 & 18.470106 \\
   5.1993378  &   9.2844258  &   12.543055 & 15.789403 & 18.668992 \\
   5.6944963  &   9.5294005  &   12.751134 & 15.955969 & 18.901905 \\
   6.1177631  &   10.005086  &   13.135366 & 16.316948 & 19.236086 \\
   6.4719956  &   10.337823  &   13.385865 & 16.523806 & 19.367222 \\
 \bottomrule
  \end{tabular}
  \caption{Ordinates of the first non trivial zeros of $L(f\otimes g\otimes h,s)$}
  \label{tab:lfunzeros}
\end{table}

\bibliographystyle{plain} 
\bibliography{./refs}

\end{document}